\documentclass[a4paper,USenglish,cleveref, autoref, thm-restate]{lipics-v2021}
\usepackage{sam-barr}
\usepackage{makecell}
\usepackage[disable]{todonotes}  % enable this for making todonotes invisible
\usepackage{chngpage}

\title{Efficiently Partitioning the Edges of a 1-Planar Graph into a Planar Graph and a Forest} %TODO Please add

\titlerunning{Partitioning the Edges of a 1-Planar Graph} %TODO optional, please use if title is longer than one line

\author{Sam Barr}{University of Waterloo, Canada}{s4barr@uwaterloo.ca}{}{}%TODO mandatory, please use full name; only 1 author per \author macro; first two parameters are mandatory, other parameters can be empty. Please provide at least the name of the affiliation and the country. The full address is optional

\author{Therese Biedl}{University of Waterloo, Canada}{biedl@uwaterloo.ca}{}{Research supported by NSERC.}

\authorrunning{S. Barr and T. Biedl}
\Copyright{Sam Barr and Therese Biedl} %TODO mandatory, please use full first names. LIPIcs license is "CC-BY";  http://creativecommons.org/licenses/by/3.0/


 %TODO mandatory; please add comma-separated list of keywords

\category{} %optional, e.g. invited paper

\relatedversion{} %optional, e.g. full version hosted on arXiv, HAL, or other repository/website
\nolinenumbers %uncomment to disable line numbering

\EventEditors{John Q. Open and Joan R. Access}
\EventNoEds{2}
\EventLongTitle{42nd Conference on Very Important Topics (CVIT 2016)}
\EventShortTitle{CVIT 2016}
\EventAcronym{CVIT}
\EventYear{2016}
\EventDate{December 24--27, 2016}
\EventLocation{Little Whinging, United Kingdom}
\EventLogo{}
\SeriesVolume{42}
\ArticleNo{23}
\hideLIPIcs
\begin{document}

\maketitle

\begin{abstract}
    1-planar graphs are graphs that can be drawn in the plane such that
    any edge intersects with at most one other edge.
    Ackerman showed that the edges of a 1-planar graph can be partitioned
    into a planar graph and a forest, and claims that the proof leads to
    a linear time algorithm. However, it is not clear how one would
    obtain such an algorithm from his proof.
    In this paper, we first reprove
    Ackerman's result (in fact, we prove a slightly more general statement)
    and then show that the split can be found in linear time by using
    an edge-contraction data structure by Holm et al.
\end{abstract}

\section{Introduction}

In this paper, we study the class of 1-planar graphs; graphs that can be drawn
in the plane such that every edge crosses at most one other edge. 1-planar
graphs were introduced by Ringel~\cite{ringel}, motivated by the problem
of coloring a planar graph and its faces.
Since then, there have been many publications concerning 1-planar graphs,
both for theoretical results such as coloring and algorithmic results such as
solving drawing and optimization problems. See~\cite{annotated-bibliography}
for an annotated bibliography from 2017
and~\cite{beyond-planar-graphs} for a more recent book that includes
developments since then.

%\todo{FYI: I much expanded the motivation for edge partition}
Many of the results for 1-planar graphs are obtained by first converting
the 1-planar graph $G$ into a planar graph $G'$, applying results for
planar graphs, and then expanding the result from $G'$ back to $G$.
(We will give specific examples below.)  There are numerous ways of 
how to create $G'$, e.g., by deleting vertices or by replacing crossings
by dummy-vertices.  Of particular interest to us here is to make a
1-planar graph planar
by deleting edges.  Put differently, we want an {\em edge partition}, i.e.,
write $E(G)= E'\cup E''$ such that $E'$ forms a planar graph while
$E''$ has some special structure that makes it possible to expand a
solution for $G'$ to one for $G$.  
%for 1-planar graphs and other classes of near-planar
%graphs is how they may be converted into planar graphs. Many have taken
%to see how the edges of these graphs might be partitioned to induce planar
%graphs.

%\medskip\noindent\textbf{Previous Work.}
\subparagraph{Previous Work.}
The main focus of this paper is a result by Ackerman~\cite{2014-Ackerman-1PlanarSplit}.
He established that
the edges of a 1-planar graph can be partitioned such that one partition induces
a planar graph and the other induces a forest.
This was an extension of an earlier result from Czap and
Hud\'{a}k~\cite{czap-1-planar-split}, who proved it for optimal 1-planar graphs
(1-planar graphs with the maximum $4n-8$ edges). 
%Their proof leads to a linear time algorithm.
Other partitions of near-planar graphs have also been studied; we list a few here.
Lenhart et al.~\cite{lenhart-1-planar-partition} show that optimal 1-planar
graphs can be partitioned
into a maximal planar graph and a planar graph of maximum
degree four 
%Moreover, they show that the partition can be found in linear time
(the bound of four is shown to be optimal).
Bekos et al.~\cite{2-plane-partition} provide
edge partition results for some $k$-planar graphs (graphs that
can be drawn in the plane such that any edge crosses at most $k$ other edges).
%and they prove linear time for these results.
Di Giacomo et al.~\cite{nic-edge-partitions} prove edge partition
results for so-called NIC-graphs, a subclass of 1-planar graphs.

%\todo{FYI: For a bigger ``wow''-effect, I've separated out the discussion 
%what's linear and what's not to cover all papers at once.}
For algorithmic purposes, we need to find such edge partitions in linear time.
With one exception, the above papers either explicitly come with a linear-time
algorithm to find the edge-partition, or such an algorithm can easily be derived
from the proof.  The one exception is the paper by Ackerman~\cite{2014-Ackerman-1PlanarSplit}. 
He claims that the partition of a 1-planar graph into a planar
graph and a forest can be found in linear time, but provides no proof.
Moreover, it is not clear how one would achieve linear time
(or even $O(n \log n)$ time),  since he relies on contracting
edges while repeatedly testing whether two vertices share a face and occasionally
splitting the graph into subgraphs, and neither of these operations can 
trivially be done in constant time in planar graphs.  (We confirmed this in private communication 
with Ackerman.)

%\medskip\noindent\textbf{Our Results.}
\subparagraph{Our Results.}
In this paper, we show that a partition
of a 1-planar graph into a planar graph and a forest can be found in linear time.
We were not able to use Ackerman's proof for this directly, so as a first step we
re-prove the result in a slightly different way to avoid some problematic
situations and so that then a linear-time algorithm
can be established. A crucial ingredient for this is a
data structure
by Holm, Italiano, Karczmarz, \L\k{a}cki, Rotenberg and Sankowski~\cite{2017-ESA-PlanarGraphContraction} 
%\todo{FYI: My postdoc supervisor (whose last name starts with `W') drilled into me that
%one should list {\em all} co-authors at least once, to make those who come later in
%the alphabet happier.}
which allows for efficiently contracting edges of planar graphs. To our knowledge,
this data structure has not been implemented.
Because of this, we also show that the partition can be computed in
$O(n \log n)$ time using a simple data structures based on incidence lists.

As a consequence of our result, a number of related problems can
be solved in linear time:
%\todo[inline]{FYI: This whole section got expanded.}
\begin{itemize}
    \item Angelini et al.~\cite{angelini2019quasefe} studied the problem
        of finding simultaneous quasi-planar drawings of graphs where some edges are fixed.
        They used Ackerman's partition result and cited its claimed linear runtime;
        with our result the linear runtime of \cite{angelini2019quasefe} is established.
    \item It is known that every 1-planar graph has {\em arboricity} 4, i.e., its 
	%\todo{FYI: arboricity wasn't defined}
	edges can be partitioned into 4 forests.  (This follows from Nash-Williams
	formula for arboricity \cite{nash-williams}
    since 1-planar graphs have at most $4n-8$ edges~\cite{schumacher-4n}.)
	The arboricity (and the corresponding edge-partition) of a graph can be computed 
    in polynomial time~\cite{gabow-arboricity},
	but to our knowledge not in linear time.  For \emph{planar} graphs, a
        split into 3 forests can be found in linear time~\cite{schnyder}.
	So with our result, the partition of a 1-planar graph into 4 forests can
        be done in linear time as well, by first partitioning into a forest and a planar
	graph and then applying \cite{schnyder} onto the planar graph.
    \item If $G$ is a bipartite 1-planar graph, then it has at most $3n-8$ edges \cite{2014-Karpov} and
    %\todo{TB: It's generally a good idea to have read at least the first 1-2 pages of {\em every} paper that we cite.  I have done this for most (though Karpov's is new to me too), but you probably should.  I have created a folder PDFofReferences; could you dump all PDFs in there and then take at least a superficial look at those you don't know already?}
	hence arboricity 3. With our result we can
        partition it into a forest and a planar bipartite graph in linear time.
        The planar bipartite graph can be partitioned into two forests in
        linear time~\cite{bipartite-2-trees}. In consequence, every 1-planar bipartite graph
        can be partitioned into three forests in linear time.
    \item From a partition into $d$ forests one easily obtains an edge orientation with
	in-degree at most $d$. For bipartite graphs, such an edge-orientation can
	be used to prove $(d{+}1)$-list-colorability~\cite{choice-number-in-graphs},
        and this list-coloring can be found in linear time~\cite{owen-manuscript}.
	Putting everything together, therefore our paper fills the one missing gap
	to show that 1-planar bipartite graphs can be 4-list-colored in linear time.
\end{itemize}

Our paper is structured as follows: In Section~\ref{definitions}
we go over necessary terminology for multigraphs and present Ackerman's
proof in order to demonstrate that it does not immediately lead to
a linear time algorithm. In Section~\ref{existence} we present our new proof.
In Section~\ref{algorithm} we use our alternate proof to design an efficient
algorithm for finding the partition, before concluding in Section~\ref{conclusion}.  Some proofs and pseudo-codes have been deferred to the appendix for space-reasons.

\section{Background}
\label{definitions}

We assume basic familiarity with graph theory (see e.g. \cite{diestel}).
All graphs in this paper are finite and connected, but not necessarily simple.

For a (multi)graph $G$ and a vertex $x$ of $G$, $d(x)$ is the number of edges
incident to $x$.

We recall that a (multi)graph $G$ is called \emph{planar} if it can be
drawn in the plane without edges crossing.
Let $G$ be a planar (multi)graph given with a drawing $\Gamma$. The maximal
regions of $\R^2 \setminus \Gamma$ are the \emph{faces} of $G$.
We add a \emph{chord} to a face $f$ by adding an edge between two non-adjacent
vertices on the boundary of $f$, and drawing the edge through the region of $f$.
For each vertex $x$
with incident edges $e_1, \ldots, e_k$, drawing $\Gamma$ places these edges in some
rotational clockwise order around $x$. The space between two edges which are adjacent
in this rotational order form an \emph{angle}. The \emph{degree} of a face $f$
is the number of angles contained in the face.  We say that a face $f$
is a \emph{quadrangle} if it has degree 4. Note that this includes
both faces with 4 vertices on their boundary, and some faces with
fewer than 4 vertices on their boundary (see Figure~\ref{contract:z1:z3}).
If $f$ is a quadrangle with exactly 4 vertices on its boundary, then we say
that $f$ is a \emph{simple quadrangle} (the quadrangles of a simple planar graph
will all be simple quadrangles).
A face of degree 3 is a \emph{triangle}, and a face of degree 2 is a \emph{bigon}.

%\todo{FYI: Moved properties of quadrangles closed to definition of quadrangles.}
The \emph{facial cycle} of a quadrangle $f$ is a 4-tuple
$\angles{z_0, z_1, z_2, z_3}$ such that each $z_i$ is on the boundary of $f$
and there are edges $z_iz_{i+1}$ and $z_iz_{i-1}$ (arithmetic modulo 4)
which form an angle in $f$. Note that $z_0,z_1,z_2,z_3$ need not be distinct
if $f$ has loops or parallel edges.
We say that $z_0$ and $z_2$ are \emph{opposing vertices in $f$} (we will often omit
mention of the face when it is clear from context).
Likewise $z_1$ and $z_3$ are opposing vertices in $f$. Note that it is possible
for a vertex to oppose itself in a quadrangle.

\emph{Stellating} a face $f$ of a planar graph is the process of adding a new vertex $s$
inside of $f$, and adding an edge from $s$ to every vertex on the boundary of $f$.
Given two vertices $a$ and $b$, we \emph{contract} $a$ and $b$ by
creating a new vertex $c$, adding an edge $vc$ for each edge $va$ and $vb$, and
deleting $a$ and $b$ and all their incident edges. If $a$ and $b$ were incident,
then $c$ has a loop, and if $a$ and $b$ were both adjacent to a vertex $v$,
then $c$ will have parallel edges to $v$. If $a$ and $b$ were both on the
boundary of some face $f$ in a planar graph, then we can
\emph{contract $a$ and $b$ through $f$} by placing this new vertex $c$
inside of $f$. This preserves planarity, but destroys the face $f$.

A \emph{1-planar} graph is a graph $G$ that can be drawn in the plane such that any
edge intersects with at most one other edge.   
Here ``drawing'' always means a \emph{good drawing} (see e.g.~\cite{2013-Schaefer}),
%\todo{FYI: I said to cite Schaefer's book for this, but I realized that that was incorrect; it's a survey-article of his that reviews the numerous variants of how to define `drawing' in detail. Citation changed.}
in particular that no edge intersects itself, and that incident edges do not cross.
From now on, whenever we speak of a 1-planar graph, we assume that one particular
1-planar drawing has been fixed.  
%\todo{FYI: Note that we must fix the drawing of the graph, else a lot of the things
%you defined don't make sense.  In the literature this is more commonly called a
%\emph{1-plane} graph, but since we don't consider graphs where the drawing varies,
%we probably can stick with ``1-planar''.}
A pair of edges that intersect
each other are a \emph{crossing pair}, and the point where they intersect
is known as the \emph{crossing point}.
The \emph{planarization} is the graph $G^\times$ obtained
by replacing each crossing point with a new vertex. 
A 1-planar graph $G$ is \emph{planar-maximal} if no planar edge can be
added to the fixed drawing of $G$. 

For algorithmic purposes, a drawing of a planar graph can be specified by
giving the rotational clockwise order of edges at every vertex; this specifies the
circuits bounding the faces uniquely.  A drawing of a 1-planar graph can
be specified by giving a drawing of its planarization, with the vertices
resulting from a crossing point marked as such.
Since testing 1-planarity is NP-hard~\cite{1-planar-np-hard}, we assume
that any 1-planar graph $G$ is given with such a drawing.  We also assume
that $G$ is planar-maximal, because any 1-planar graph can be made planar-maximal
in linear time by adding edges~\cite{maximal-1-planar-linear},
and having more edges can only make partitioning more difficult.

For ease of notation, we define a shortcut for our partition problem.

\begin{definition}
    A graph $G$ has a \emph{PGF-partition} if $E(G)$ can be partitioned into
    two sets $A$ and $B$ such that $G[A]$ is a planar graph and $G[B]$ is a forest.
\end{definition}

\subsection{Ackerman's Proof}

To establish the difficulties of achieving a linear time algorithm, we briefly
review here Ackerman's proof of the existence of a PGF-partition.

%If we were to remove all of the crossing pairs from a 1-planar graph, we would
%have a planar graph. Thus, one needs only to select one edge from each crossing pair
%such that the selected edges induce forest. Moreover, since any edges which
%form a crossing pair are not
%incident, Ackerman observes that this is equivalent to adding a chord to
%simple quadrangles of a planar graph such that the chords induce a forest.

Let $G$ be a (planar-maximal) 1-planar graph without loops drawn in the plane.
Remove all crossing pairs of $G$. Call the resulting graph $H$ the
\emph{(planar) skeleton} of $G$ 
%(this graph is also referred to as the \emph{planar skeleton}~
\cite{maximal-1-planar-linear}.
%\todo{SB: I removed the reference to red graph; upon looking back at the literature
%the usage of this term was more vague than I initially thought and mostly
%just refers to red-blue edge colorings. In light of this we could just switch to
%using planar skeleton, but we already have a few definitions with the term
%planar in them.  TB: I've reworded it so that we use `planar' only once.}
Observe that the faces of $H$ are either bigons, triangles, or quadrangles, and that
there is a 1-1 mapping between the quadrangles of $H$ and the crossing pairs of $G$.
Moreover, by this 1-1 mapping and the fact that the two edges forming a crossing pair
do not share an endpoint, one can see that the quadrangles of $H$ are in
fact simple quadrangles. Ackerman, similarly to Czap and Hud\'ak \cite{czap-1-planar-split},
establishes the following.

\begin{lemma}[Ackerman \cite{2014-Ackerman-1PlanarSplit}]
    \label{adding:chords:lemma}
    Let $G$ be a 1-planar graph, and let $H$ be the skeleton of $G$.
    If we can add a chord to every quadrangle of $H$ such that the chords induce
    a forest, then $G$ has a PGF-partition.
\end{lemma}
%\begin{proof}
%    Let $C$ be the set of chords added to $H$. By the 1-1 mapping between
%    quadrangles of $H$ and crossing pairs of $G$, we know that exactly one
%    edge from each crossing pair of $G$ is contained in $C$. In particular,
%    each edge $e \in C$ forms a crossing pair with some edge $e'$ of $G$.
%    Let $C'$ be the set of these edges $e'$.
%    By assumption $G[C]$ is a forest, and moreover $H \cup C'$ is planar,
%    and we have the desired partition
%\end{proof}

Thus, in order to prove the existence of a PGF-partition,
it suffices to show (typically by induction on the number of quadrangles)
that such a set of chords can be found.
For the induction to go through, Ackerman additionally forbids
the chords from containing a path
between two adjacent pre-specified vertices $x,y$.

\begin{theorem}[Ackerman \cite{2014-Ackerman-1PlanarSplit}]
    Let $H$ be a plane multigraph without loops such that every face has degree
    at most four and all quadrangles are simple.
    Let $x,y$ be a pair of adjacent vertices of $H$.
    Then we can add a chord to every simple quadrangle of $H$ such that the
    subgraph induced by the chords is a forest and
    does not contain a path between $x$ and $y$.
\end{theorem}
\begin{proof}
    Proceed by induction on the number of quadrangles in $H$.
    If $H$ has no quadrangles, then the statement is trivial. Otherwise, let $f$ be
    a quadrangle with facial cycle $\angles{z_0, z_1, z_2, z_3}$; by assumption
    $f$ is simple.

    \textbf{Case 1:} The only face containing $z_0$ and $z_2$ is $f$; in
    particular $z_0$ and $z_2$ are not adjacent. Contract $z_0$ and $z_2$
    through $f$. Let $H'$ be the graph resulting from this contraction. Observe that
    $H'$ has one fewer quadrangle than $H$. All other quadrangles remain
    simple since $f$ was the only face containing $z_0$ and $z_2$.
    Apply induction on $H'$ with the same pair of adjacent vertices $x,y$ to receive
    a set of chords $C'$, and then further add the chord $z_0z_2$.
    Chords have now been added to every quadrangle of $H$, and it is easy
    to see that we have not added a cycle or a path from $x$ to $y$ in the chords.

    \textbf{Case 2:} There is some face $f'\neq f$ containing $z_0$ and $z_2$,
    but the only face containing $z_1$ and $z_3$ is $f$. Proceed as in
    Case~1, except contract $z_1$ and $z_3$.

    \textbf{Case 3:} None of the above. Then there is a face $f' \neq f$ containing
    $z_0$ and $z_2$, and there is a face $f'' \neq f$ containing $z_1$ and $z_3$.
    Ackerman argues that $f' = f''$ (see also Lemma~\ref{quadrangle:lemma:2}),
    and therefore $f'$ is also a simple quadrangle.
    Observe that $G$ can be split into four connected subgraphs $H_0,H_1,H_2,H_3$,
    where $H_i$ contains $z_i$ and $z_{i+1}$ (addition modulo 4) on the boundary
    (see Figure~\ref{contract:z1:z3}(e)).
    One of these subgraphs, say $H_0$, will contain the adjacent pair $x,y$.
    Apply induction on $H_0$ with $x,y$, and apply induction on the other $H_i$ with
    the pair $z_i,z_{i+1}$. After induction, add the chord $z_0z_2$ in $f$,
    and $z_1z_3$ in $f'$. One verifies that the added chords to not add
    a path between $x$ and $y$ and that the chords to not create a cycle.
\end{proof}

\section{An Alternate Existence Proof}
\label{existence}

While Ackerman's proof clearly leads to a polynomial time algorithm
for finding the partition, it is not obviously linear since
distinguishing between the cases and contracting are not obviously doable
in constant time:
\begin{enumerate}
    \item We need to test whether a given pair of vertices share more than one face.
    \item The graph changes via contractions, and it is not obvious
        whether the existing data structures for efficiently contracting edges in planar graphs
        (e.g.~\cite{2017-ESA-PlanarGraphContraction}) would support (1)
        in constant time.
    \item In Case~3 of Ackerman's proof, we need to identify the four
        subgraphs $H_0,H_1,H_2,H_3$.  Furthermore, we need to determine which of these
        subgraphs contains the pair $x,y$.  Neither operation is obviously
	doable in constant time. 
\end{enumerate}

We now give a different proof of the existence of a PGF-partition that
either avoids these issues or addresses explicitly how to resolve them.
The biggest change is how we handle Case~3 of Ackerman's proof. Ackerman used
here a split into four graphs, which is necessary in order to maintain that
all quadrangles are simple. We prove a more general statement that
permits non-simple quadrangles and hence avoids having to split the graph. Moreover, we generalize Ackerman's
\say{forbidden pair} $x$ and $y$ by choosing chords in such a way that chords
do not induce a path between {\em any} pair of vertices that were adjacent
in the original graph $H$.
Doing so simplifies the induction since we no longer need to keep track
of where the vertices $x$ and $y$ are.
Before we state this result 
%in Theorem~\ref{main_theorem} below, 
% TB: no need to say this since now the result is very close
we need
a few helper-results that hold for all quadrangles (simple or not).
They are easily proved by finding (after minor modifications) a planar drawing of $K_5$
if the conclusion is violated; see the appendix for details.

\begin{lemma}
    \label{quadrangle:lemma:1}
    Let $H$ be a plane multigraph without loops, let $f$ be a quadrangle of $H$,
    and let $\angles{z_0,z_1,z_2,z_3}$ be the facial cycle of $f$.
    If $z_i=z_{i+2}$ (addition module 4) for some $i$,
    then $z_{i+1} \neq z_{i+3}$, and there is no face $f'\neq f$ which contains
    $z_{i+1}$ and $z_{i+3}$.
\end{lemma}
\iffalse
\begin{proof}
    (see Figure~\ref{contract:z1:z3}(b))
    Up to renaming, we may assume that $i=0$, so $z_0=z_2$.
    Assume by way of contradiction that $z_1=z_3$.
    Then $f$ consists of several parallel edges between $z_0=z_2$ and $z_1=z_3$, and thus
    $f$ is not a quadrangle.

    Assume by way of contradiction that there is some face $f' \neq f$ which
    contains $z_1$ and $z_3$. Subdivide one of the $z_0z_3$ edges to obtain
    a new vertex $y$ adjacent to $z_0$ and $z_3$, and further add an edge $yz_1$
    within $f$.
    We have split $f$ and attained a simple quadrangle $f''$ with facial cycle
    $\angles{z_0, z_1, y, z_3}$. Stellate $f''$ with a new
    vertex $c$, and add an edge $z_1z_3$ through the face $f'$.
    All these steps maintain the planarity of $H$. Moreover, the five
    vertices $z_0, z_1, y, z_3, c$ are pairwise adjacent. But this forms a $K_5$ which
    is not planar, a contradiction.
\end{proof}
\fi

The following lemma was shown (without being stated explicitly) in
Case~3 of Ackerman's proof.

\begin{lemma}
    \label{quadrangle:lemma:2}
    Let $H$ be a plane multigraph without loops, let $f$ be a quadrangle of $H$,
    and let $\angles{z_0,z_1,z_2,z_3}$ be the facial cycle of $f$.
    If $z_i$ and $z_{i+2}$ (addition modulo 4) for some $i$ are both on some face
    $f' \neq f$, then no face $f'' \neq f,f'$ contains both $z_{i+1}$ and $z_{i+3}$.
\end{lemma}
\iffalse
\begin{proof}
    (see Figure~\ref{contract:z1:z3}(a,c,d,e))
    Up to renaming we may assume that $i=0$, so $z_0$ and $z_2$
    are on $f$ and $f'$. Suppose by contradiction that such a face $f''$ exists.
    By Lemma~\ref{quadrangle:lemma:1}, $z_1 \neq z_3$ and $z_0 \neq z_2$.
    Stellate $f$ with a new vertex $c$, add an edge $z_0z_2$ through $f'$,
    and add an edge $z_1z_3$ through $f''$. The original multigraph $H$ was planar,
    and all of these operations preserve planarity. However, the five vertices
    $z_0, z_1, z_2, z_3$, and $c$ are pairwise adjacent and form a $K_5$, which
    is not planar, a contradiction.
\end{proof}
\fi

\begin{figure}[ht]
    \centering
    \includegraphics[scale=0.65,page=1]{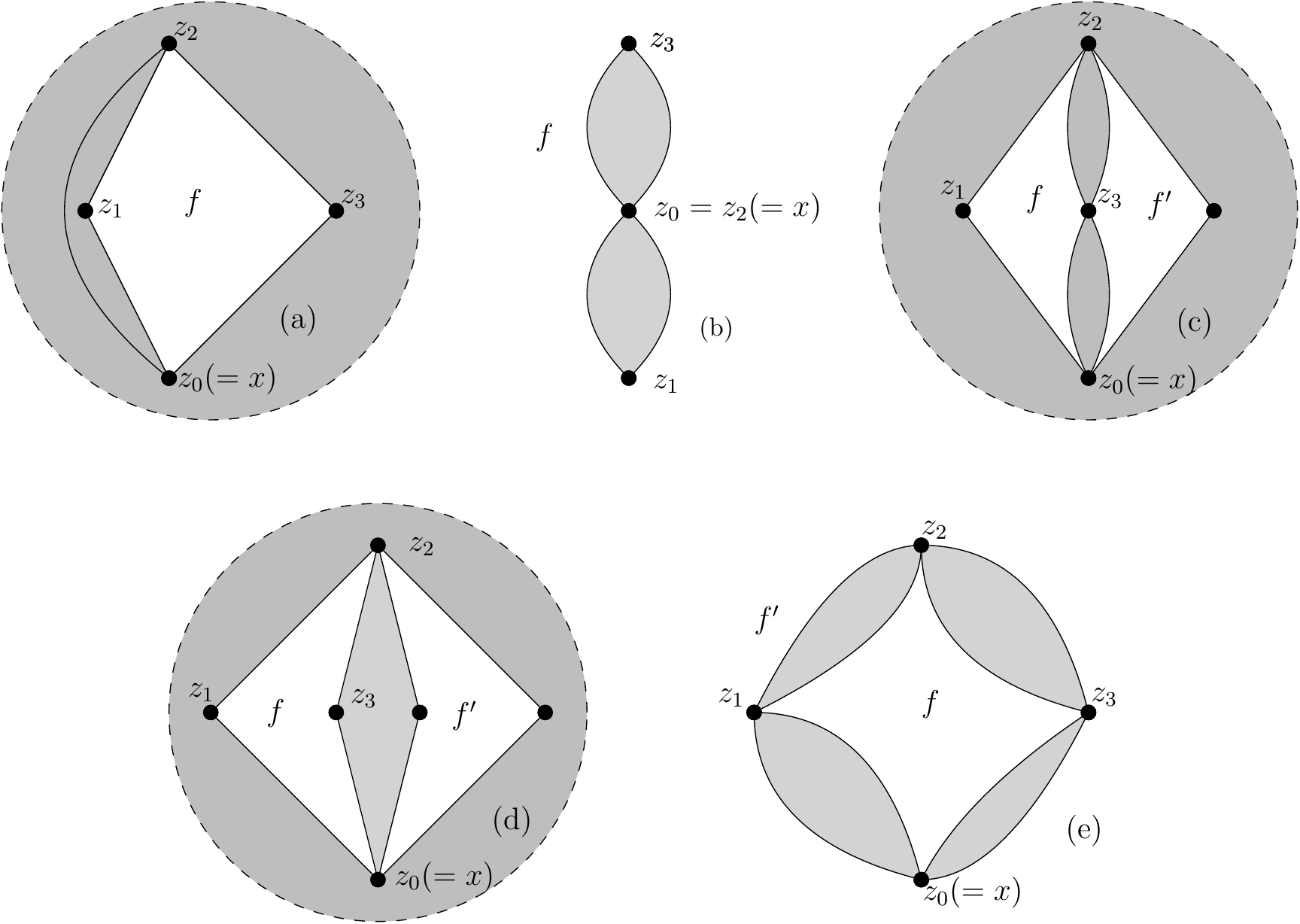}
    \caption{Some configurations where we contract $z_1$ and $z_3$}
    \label{contract:z1:z3}
\end{figure}

Now we reprove the existence of quadrangle-chords that form a forest.

\begin{theorem}
    \label{main_theorem}
    Let $H$ be a plane multigraph without loops such that every face has
    degree at most 4. Then it is possible to add a chord to every quadrangle of $H$
    such that the graph induced by the chords is a forest. Moreover,
    for any pair $a$ and $b$ of adjacent vertices of $H$, there is no
    path from $a$ to $b$ in the chords.
\end{theorem}
\begin{proof}
    As in Ackerman's proof, we prove the claim by induction on the number of
    quadrangles in $H$ and remove each quadrangle by contracting an opposing
    pair of vertices in the quadrangle. If $H$ has no quadrangles, the claim is trivial.
    Otherwise, there are quadrangles left. Pick an arbitrary vertex $x$ that
    is still incident to some quadrangles (later in our algorithm we will
    deal with all its incident quadrangles). Let $f$ be one of its incident
    quadrangles with facial cycle $\angles{x=z_0, z_1, z_2, z_3}$.
    We first pick two opposing vertices of $f$ to contract.

\medskip\noindent\textsf{\textbf{Case 1:}} This case covers when we choose to contract $z_1$ and $z_3$,
    and has three sub-cases.  We contract $z_1$ and $z_3$ whenever
    \begin{description}
        \item[Case 1.a] $x=z_2$ are the same vertex, or
        \item[Case 1.b] $x$ and $z_2$ are adjacent, or
        \item[Case 1.c] $x$ and $z_2$ are opposing vertices of some quadrangle $f' \neq f$.
    \end{description}

    Figure~\ref{contract:z1:z3} illustrates possible configurations of face $f$
    where Case~1 applies.
    Note that Case~1.c covers Case~3 of Ackerman's proof, where
    $x,z_1,z_2,z_3$ all belong to two simple quadrangles $f,f'$
    (see also Figure~\ref{contract:z1:z3}(e)). Our contraction turns $f'$ into
    a non-simple quadrangle, but our proof can handle this.

\medskip\noindent\textsf{\textbf{Case 2:}} Otherwise, we contract $x$ and $z_2$.

\medskip
    Table~\ref{my:table} demonstrates when we pick Case~1 and when we pick Case~2,
    and crucially shows cases which are impossible by Lemmas~\ref{quadrangle:lemma:1}
    and~\ref{quadrangle:lemma:2}. To see that these lemmas apply in second row
    and column, observe that adjacent vertices always share face, and in particular
    if two opposing vertices are adjacent then they must share a face other than
    the quadrangle they are opposing in.

\begin{table}[ht]  \centering
% TB: In general, "centering" is preferred over "begin{center}" inside figures and tables else the spacing is weird.
%    \scriptsize  % doesn't seem needed
%    \begin{adjustwidth}{-.5in}{-.5in}  % not sure what this does, so removed it
%    \begin{center}
    \begin{tabular}{c|c|c|c|c}
         & $x=z_2$
         & \makecell{$x \neq z_2$; \\ $(x,z_2) \in E(H)$}
         & \makecell{$x \neq z_2$; \\ $x$, $z_2$ are \\ opposing in $f'$}
         & Otherwise
        \\ \hline
        $z_1=z_3$ & \makecell{Impossible \\ (Lemma~\ref{quadrangle:lemma:1})}
          & \makecell{Impossible \\ (Lemma~\ref{quadrangle:lemma:1})}
          & \makecell{Impossible \\ (Lemma~\ref{quadrangle:lemma:1})} & Case 2
        \\ \hline
        \makecell{$z_1\neq z_3$; \\  $(z_1,z_3)\in E(H)$}
          & \makecell{Impossible \\ (Lemma~\ref{quadrangle:lemma:1})}
          & \makecell{Impossible \\ (Lemma~\ref{quadrangle:lemma:2})}
          & \makecell{Impossible \\ (Lemma~\ref{quadrangle:lemma:2})} & Case 2
        \\ \hline
        \makecell{$z_1 \neq z_3$; \\ $z_1$, $z_3$ are \\ opposing in $f''$}
          & \makecell{Impossible \\ (Lemma~\ref{quadrangle:lemma:1})}
          & \makecell{Impossible \\ (Lemma~\ref{quadrangle:lemma:2})}
          & \makecell{
              Impossible if $f' \neq f''$ \\
              (Lemma~\ref{quadrangle:lemma:2}), \\
              Case 1.c otherwise
          }
          & Case 2
        \\ \hline
        \rule{0pt}{12pt} % last row is too short otherwise
        Otherwise & Case 1.a & Case 1.b & Case 1.c & Case 2
    \end{tabular}
%    \end{center}
%    \end{adjustwidth}
    \caption{All possible cases for the quadrangle $f$ with facial cycle
        $\angles{x,z_1,z_2,z_3}$.  We either indicate
        which case in the proof of Theorem~\ref{main_theorem} would be
        chosen, or indicate the lemma that demonstrates that this case is impossible.
%\todo[inline]{Changed Case 2 to Case 1 in (3,3), please double-check.}
    }
    \label{my:table}
\end{table}

    Let $z_i, z_{i+2}$ be two vertices chosen for contraction and
    let $H'$ be the graph resulting from the contracting $z_i$ and $z_{i+2}$. By
    Table~\ref{my:table}, $z_i$ and $z_{i+2}$ are not adjacent, and
    they are unique.  Therefore our
    contraction has destroyed the quadrangle $f$ and $H'$ has no loops, so we
%\todo{FYI: We sometimes write ``loop'' and sometimes ``self loop''.  Should stick with one; I used ``loop''.}
    can apply induction on $H'$. Let $C'$ be the set of chords added to $H'$.
    By the inductive hypothesis, $C'$ induces a forest and for any edge $ab \in H'$,
    there is no path from $a$ to $b$ in $C'$.
    Uncontract $z_i$ and $z_{i+2}$, and add a chord $e:=z_iz_{i+2}$
    between them. Define $C := C' \cup \{e\}$. We
    now verify that $C$ satisfies all conditions.

    Let $a,b$ be a pair of adjacent vertices of $H$. Assume by way of contradiction
    that there is an path from $a$ to $b$ in $C$. By the inductive hypothesis,
    the path must use $e$. Furthermore, $e$ cannot be the edge $ab$
    since $z_i$ and $z_{i+2}$ are not adjacent,
    so the path must use some edges from $C'$. Let
    $c_1, \ldots, c_{k_1}, e, c_{k_1+1}, \ldots, c_{k_2}$ be this path, $k_2 \geq 1$.
    But then $c_1, \ldots, c_{k_2}$ would be a path from $a$ to $b$ within
    $C'$ in $H'=H/e$, a contradiction.

    Assume by way of contradiction that $C$ induces a cycle in $H$.
    Since $e$ is not a loop in $H$, the cycle must use edges from $C'$. Let
    $e, c_1, \ldots, c_k$ be the cycle, $k \geq 2$.  Then $c_1, \ldots, c_k$
    would induce a cycle within $C'$ in $H'=H/e$, a contradiction.
\end{proof}

\section{Efficient Implementation}
\label{algorithm}

It is still not immediately clear how one would implement the above theorem
in order to achieve linear runtime, since as in Ackerman's proof we need to
repeatedly test how many quadrangles two vertices share. However, if
we are more careful about the order in which we contract each quadrangle,
an efficient implementation can be achieved.
The crucial idea will be to pick some vertex $x$ and contract all quadrangles
incident to $x$. This will allow us to store additional information relative
to $x$ and hence speed up testing which case applies.

%\subsection{Data Structure Requirements}
\subsection{Data Structure Interface}

As mentioned earlier, one of the major ingredients to achieve fast run-time
%\todo{FYI: Made Holm even more prominent}
is to use the data structure by Holm et al.~\cite{2017-ESA-PlanarGraphContraction}
for contraction in planar graphs, but we will also provide a
(simpler but slower) alternative.  We will discuss these later 
(in Subsection~\ref{contract:runtime}) when we
analyze the run-time, but note here the two operations provided
by \cite{2017-ESA-PlanarGraphContraction} that will be needed:
\begin{itemize}
    \item $x = \texttt{contract}(e)$ takes a reference to an edge $e$, 
        contracts $e$, and returns the vertex resulting from the contraction.

	Note that contracting $e$ creates a loop in the graph, especially
	if there are multiple copies of this edge, while the proof of
	Theorem~\ref{main_theorem} assumed that the graph has no loops.    We could remove 
	loops (the data structure by Holm et al.~can report newly created loops 
	after \texttt{contract}), but this turns out to be unnecessary:  We will
	only contract edges at artificial gadgets inserted into the graph, and
	the created loops are at quadrangles that are destroyed afterwards and
	those will not pose problems.
    \item \texttt{neighbors}$(x)$ returns an iterator over
        $\setb{\angles{xv, v}}{xv \in E}$ where $E$ is the edge set
        of the graph, i.e. tuples
        containing each edge incident to $x$ and the endpoint of this edge.
        No guarantee is given as to the order of the neighbors.
        \todo{SB: we haven't defined $H$ or $H^\diamondsuit$ yet, removed
        references to them in this section.}

	We assume that $\texttt{neighbors}(x)$ has $O(1)$ runtime and that
	the returned list can be iterated over in $O(d(x))$ time (recall that
	$d(x)$ denotes the number of edges incident to $x$). 
	Since edge-contraction can 
	create parallel edges, it is possible that $\texttt{neighbors}(x)$ contains
	parallel edges, and hence the second element of the tuple need not be unique.
	Again this will not pose problems later.
\end{itemize}
%We will discuss the runtime of $\texttt{contract}(e)$ in
%Subsection~\ref{contract:runtime}, as it will be determined by the specific data
%structure.

In Subsection~\ref{subsection:modifications}, we will add labels and other meta-data
to vertices of our graph.  We make no assumptions as to how the meta-data %and labels
are updated when two vertices are contracted, and so we will maintain those manually.

\subsection{Preprocessing}% and an Outline}
% TB: The outline has mostly been moved to higher, so don't say this anymore
\label{subsection:modifications}

We take as input a simple 1-planar graph $G$, given by specifying its planarization
via the rotational clockwise orders of edges at the vertices, and assuming that 
vertices of the planarization that corresponds to crossings are marked as such.
From $G$, we can construct a planar-maximal supergraph $G^+$ in linear time 
\cite{maximal-1-planar-linear}, and along the way construct the planarization
$(G^+)^\times$ of $G^+$.   As before we use $H$ to denote the
skeleton of $G^+$, but we do not construct it explicitly.
%\todo{FYI: Previously the text made it sound as if we construct $H$ and then
%insert the gadgets.  But the quad-vertices are {\em exactly} the ones that were
%already in the planarization, so reworded this quite a bit.}
Instead, notice that the vertices of $(G^+)^\times$ marked as crossings
correspond uniquely to the quadrangles of $H$.  For this reason, we assume
that these vertices are marked with a label $\qu$ and we call such a vertex a
\emph{quadrangle-vertex} (whereas the corresponding face of $H$ is called
a \emph{quadrangle-face}).

Our proof of Theorem~\ref{main_theorem} relies heavily on having faces,
while the data 
structure of Holm et al.~makes no provisions for accessing faces.  For
this reason, we keep the quadrangle-vertices in the graph as representatives
of the quadrangle-faces. This will make it possible to implement the
operation of ``contract $z_i$ and $z_{i+2}$ within quadrangle $f$''
used in Theorem~\ref{main_theorem}
via edge-contractions at the corresponding quadrangle-vertex $f$ 
(see Procedure~\ref{contract} in the appendix for details).    
We also assume that quadrangle-vertex $f$ has references to the
two original edges in $G$ that crossed; when doing such a contraction
within $f$ we can hence also record the corresponding edge $z_iz_{i+2}$
for inclusion in the forest-part of the partition.
%Recall that the quadrangle-face $f$ existed in $H$ because $G$ had two crossing
%edges $z_0z_2$ and $z_1z_3$; we assume that quadrangle-vertex $f$ stores
%references to these edges of $G$.

Our proof of Theorem~\ref{main_theorem} also requires
knowing the order of vertices along a quadrangle, and the data structures
do not support this directly.  Therefore at any quadrangle-vertex $f$ we stellate
each of the four incident triangular faces; 
see Figure~\ref{gadget}.
%\todo{FYI: Shortened how we describe the gadget, with more emphasis on why we can do the labels}
Since we have not yet contracted any edges, we have access to the rotational
clockwise order of edges at $f$; we can hence label the added vertices with $\qu_i$
(for $0\leq i\leq 3$) in clockwise order around $f$.  With this, we can
retrieve the clockwise order $\langle z_0,\dots,z_3 \rangle$ of vertices
on the quadrangle-face corresponding to $f$ in constant time
(see Procedure~\ref{reconstruct} in the appendix for details).

%Let $f$ be a quadrangle of $H$ with facial cycle $\angles{z_0, z_1, z_2, z_3}$.
%We add five vertices $f, f_0, f_1, f_2, f_3$.
%Note that $f$ can now refer to either the quadrangle of $H$ or the vertex in the gadget;
%we will use \emph{quadrangle-face} and \emph{quadrangle-vertex} (respectively)
%to disambiguate. We add edges $fz_i$, $f_iz_i$, $f_{i+1}z_i$, $0 \leq i \leq 3$ and edges
%$ff_i$, $0 \leq i \leq 3$ (see Figure~\ref{gadget}). We add the label
%$quad$ to $f$, and to each $f_i$ we add the label $quad_i$.

% TB: moved figure down to where we need it
\begin{figure}[ht]
    \centering
    \includegraphics[scale=0.55,page=2]{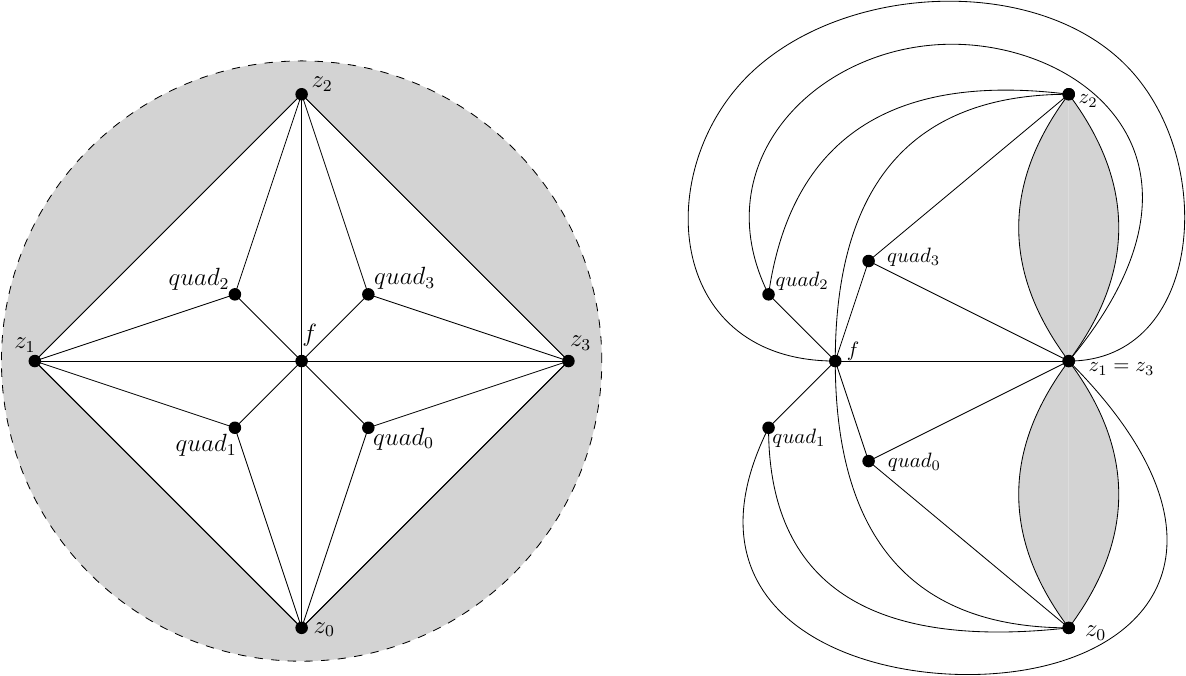}
    \caption{The gadget added to every quadrangle of $H$, shown for both a quadrangle
    with four vertices on its boundary (left)
    and one with three vertices on its boundary (right).
%\todo[inline]{Use different marks for special vertices, e.g. white circle for $f$,
%gray square for $f_0,\dots,f_3$.  Also note that $f_i$ is no longer used in text;
%replace this by $\qu_i$.  (Your ipe is too new for my computer, so I can't do this myself.)  }
}
    \label{gadget}
\end{figure}

We use $H^\diamondsuit$ for the graph that results after all these modifications
(it can be viewed as the skeleton $H$ with a ``diamond''-gadget inserted into
each quadrangle-face).  
We also add the following meta-data to each vertex of $H^\diamondsuit$:
\begin{itemize}
    \item A boolean \texttt{adj}, initialized to \textbf{false}.
    \item A boolean \texttt{in\_worklist}, initialized to \textbf{false}.
    \item An integer \texttt{opposing}, initialized to 0.
\end{itemize}
The main idea of our algorithm is to iteratively contract all the quadrangles
incident to some vertex $x$. As we do this, we will use \texttt{adj} to mark
vertices that are adjacent to $x$, \texttt{in\_worklist} to mark unprocessed
quadrangles incident to $x$, and \texttt{opposing} to keep track of the number of
quadrangle-faces where vertex $y$ is the opposing vertex of $x$.

 Since $G$ is a 1-planar graph (and therefore has at most $4n-8$ edges),
 $H^\diamondsuit$ has $O(n)$ edges. All steps in this preprocessing can
hence be done in linear time.

\subsection{Handling the Quadrangles around a Vertex}
\label{all:the:contractions}

The main subroutine of our algorithm 
handles all quadrangles incident to some vertex $x$, by contracting each
of them using the criteria
laid out in Theorem~\ref{main_theorem} to decide which vertices to contract.
To do so, we will initialize and maintain a work-list of faces incident to $x$
that we need to contract (using \texttt{in\_worklist} to avoid putting duplicate quadrangles into
the worklist).  We also mark vertices in $H^\diamondsuit$ as \texttt{opposing}
and \texttt{adj} to $x$ as needed;
this can be done in $O(d(x))$ time by retrieving all neighbours of $x$ via \texttt{neighbours}.
Procedure~\ref{alg:add:quads} in the appendix shows the details.
%we use $y$ rather than $x$ since this procedure will be used elsewhere).

%After initializing our worklist we
%set \texttt{adj} for all the neighbors of $x$.
Now we iteratively
contract all the faces in the worklist according to Theorem~\ref{main_theorem};
with \texttt{adj} and \texttt{opposing} we can determine the correct case
in constant time. See Procedure~\ref{main_algorithm} for details.
In Case~1, we need to update some values of opposing: the vertex
$z_2$ which was opposing $x$ now no longer opposes $x$ in $f$
(since $f$ was destroyed), so we decrement $z_2.\texttt{opposing}$.
Likewise the new vertex $v$
resulting from the contraction will be opposing $x$ in all those quadrangles in which
$z_1$ and $z_3$ previously opposed $x$, so we set $v.\texttt{opposing}$ correspondingly.
In Case~2, when we contract some vertex $z_2$ into $x$, we need to
add the quadrangles incident to $z_2$ to our worklist, for which we can
re-use Procedure~\ref{alg:add:quads}.

Lastly, once our worklist is empty (and hence there are
no more quadrangles incident to $x$),
we reset the meta-data of $x$ and its neighbors, so that when repeating
the procedure with a different vertex as $x$ there are no stray vertices
with meta-data set to erroneous values.

    \SetKwFunction{facialcycle}{FacialCycle}
    \SetKwFunction{contractthrough}{ContractThrough}
    \SetKwFunction{addquads}{InitializeAtOneVertex}
    \SetKwFunction{cleanup}{CleanupAtOneVertex}
    \SetKwFunction{handlequads}{HandleQuadsAtOneVertex}
    \DontPrintSemicolon
\begin{procedure}
    \KwResult{Contract all the quadrangle-faces incident to a vertex $x$}

    \texttt{worklist} := []\;
    \addquads{$x$, \textnormal{\texttt{worklist}}} \tcp*{see Proc.~\ref{alg:add:quads} in appendix}
%    \For{$v \in \texttt{neighbors}(x)$} {
%        $v.\texttt{adj} := \textbf{true}$
%    }
    \For{quadrangle-vertex $f$ in \textnormal{\texttt{worklist}}} {
        $\angles{z_0, z_1, z_2, z_3} := \facialcycle{f}$ \tcp*{see Proc.~\ref{reconstruct} in appendix}
        relabel $z_i$ such that $x$ equals $z_0$\;
        \uIf(\tcp*[f]{Case 1}){$z_2$ equals $x$ \textbf{or} $z_2$.\textnormal{\texttt{adj}} is true \textbf{or} $z_2.\textnormal{\texttt{opposing}} \geq 2$} {
            \texttt{opposing1} := $z_1$.\texttt{opposing}\;
		\texttt{opposing3} := $z_3$.\texttt{opposing}\; 
            $v$ := \contractthrough{$z_1, z_3, f$} \tcp*{see Proc.~\ref{contract} in appendix} 
            $v$.\texttt{adj} := \textbf{true}\; 
            $v$.\texttt{opposing} := \texttt{opposing1} + \texttt{opposing3}\;
            $z_2$.\texttt{opposing} $-\!=$ 1
        } 
	\Else(\tcp*[f]{Case 2}){
            \addquads{$z_2$, \textnormal{\texttt{worklist}}}\;
%            \For{$v \in \texttt{neighbors}(z_2)$} {
%                $v.\texttt{adj} := \textbf{true}$
%            }
            $x$ := \contractthrough{$x, z_2, f$}\;
        }
    }
    \cleanup{$x$} \tcp*{see Proc.~\ref{alg:cleanup} in appendix}

    \caption{HandleQuadsAtOneVertex($x$)}
    \label{main_algorithm}
\end{procedure}

\subsection{Putting it All Together}
\label{all:together:now}
\label{contract:runtime}

The following summarizes our algorithm: after preprocessing, and for as
long as there is a quadrangle-face $f$ left, process all quadrangles
at a vertex $x$ on $f$ and record all edges that belong to the forest along the
way. See Procedure~\ref{final:alg} in %Appendix~\ref{appendix-algorithms}
% TB: Nothing wrong with this, but we write "in the appendix" everywhere else and should be consistent.
the appendix
for a detailed description.

%\subsection{Runtime}
It remains to analyze the run-time.  For now, we ignore the time required to 
perform the contractions and analyze the time for handling all quads at one
vertex $x$.  Initialization takes $O(d(x))$ time, and most other steps
take constant time per handled quadrangle, with one notable exception: When we
contract some vertex $z_2$ into $x$, we must update the worklist, which takes
time $O(d(z_2))$.  Complicating matters further, $z_2$ may actually be the 
result of prior contractions, so its degree may be more than what it was in
$H^\diamondsuit$, and we must ensure that degrees of vertices are not counted
repeatedly.  

\todo{FYI: Expanded this quite a bit since I misunderstood it last time.}
To handle this,  let 
$H^\diamondsuit_f$ be the graph that results from $H^\diamondsuit$
after {\em all} quadrangle-vertices have been contracted, and let
$s(x)$ be the set of vertices that were contracted into $x$, either
directly (when handling the quadrangles at $x$) or indirectly (i.e.,
if they had been contracted into one of the vertices $z_2$ that later
get contracted into $x$).    Crucially, note that if $x,x'$ are two
vertices on which we perform \texttt{HandleQuadsAtOneVertex}, then
$s(x)$ and $s(x')$ are disjoint.  This holds because once we are done 
with the first of
them (say $x$), all vertices in $s(x)$ have been combined with $x$
and no longer have any incident quadrangles.  Since we only contract
vertices into $x'$ that are incident to quadrangles, none of the 
vertices in $s(x)$ becomes part of $s(x')$.

%Each vertex $x$ of this final graph was the result of contracting
%some group of vertices $s(x) \subseteq V(H^\diamondsuit)$.
%These sets $s(x)$ form a vertex-partition
%of $V(H^\diamondsuit)$. 
Hence for each vertex $x$ of $H^\diamondsuit_f$, the amount
of work done in Procedure~\ref{main_algorithm} is proportional to
the sum of the degrees $d_{H^\diamondsuit}(y)$ for each $y \in s(x)$.
Since $H^\diamondsuit$ has $O(n)$ edges, and all other parts of the
algorithm take constant time per quadrangle-vertex,
the total amount of work done is at most
\[
    O\left(
        \sum_{x \in V\left(H^\diamondsuit_f\right)}
        \sum_{y \in s(x)} d_{H^\diamondsuit}(y)
    \right)
    = O\left(\sum_{x \in V(H^\diamondsuit)}d_{H^\diamondsuit}(x)\right)
    = O\left(|V\left(H^\diamondsuit\right)|\right)
    = O(|V(G)|)
\]
hence the algorithm is linear (ignoring the time for contractions).

%\subsection{The Runtime for \texttt{contract}}
%\label{contract:runtime}

%Since aside from contractions our algorithm is linear, the running time
%of our algorithm will be dominated by the amortized time required to perform
%the contractions. We examine two different data structures which can be used
%to implement our algorithm and their resulting run-times.

Now we consider the run-time of possible data structures for contractions.
Our first approach is to represent the graph with incidence lists, where every
vertex has a list of incident edges, each edge knows both of its endpoints,
and every list knows its length.
We can implement \texttt{neighbors}$(x)$ in constant time by simply
returning an iterator to the incidence list at $x$.
Contracting two vertices $u$ and $v$ can be done in $O(\min\{d(u),d(v)\})$ time
by re-attaching the edges of the vertex with smaller degree to the vertex
with larger degree. As with \textsc{Union-Find} data structures
implemented with linked lists (see e.g. Section 4.6 of~\cite{kleinberg-tardos}),
one shows that the amortized time for this is $O(\log n)$ per contraction.
In particular, for graphs with linearly many edges, a set of $\Theta(n)$
contractions can be done in $O(n \log n)$ time. With this we have
our first result.

\begin{theorem}
    Let $G$ be a 1-planar graph implemented with incidence lists.
    It is possible to find a PGF-partition of $G$ in $O(n \log n)$ time.
\end{theorem}

To improve this runtime,
%let $\mathcal C(G)$ be the data structure introduced by Holm et al. for a planar
%graph $G$. Our algorithm was designed to use a strict subset of the interface
%provided by $\mathcal C(G)$, and moreover $\mathcal C(G)$ satisfies the following
%theorem:
we appeal to the following result by Holm et al.
%\todo{FYI: I didn't see a reason to give this data structure a name, so removed it.}

\begin{theorem}[Holm et al. \cite{2017-ESA-PlanarGraphContraction}]
    \label{fancy:data:structure}
    Let $G$ be a planar graph with $n$ vertices and $m$ edges.
	Then there exists a data structure %$\mathcal C(G)$ 
	that supports
	\texttt{contract} and \texttt{neighbours} and that can be
%    Then the data structure $\mathcal C(G)$ can be 
	initialized in $O(n+m)$ time.
    Any calls to \texttt{neighbors} can be processed in worst case constant time,
    and any sequence of calls to \texttt{contract} can be performed
    in time $O(n+m)$.
\end{theorem}

Since our graph has $O(n)$ edges, we have the main result of this paper.

\begin{theorem}
    Let $G$ be a 1-planar graph. Then in $O(n)$ time we can find
	an edge-partition of $G$ into a forest and a planar graph.
%A PGF-partition of $G$ can be found in $O(n)$ time.
	%\todo{FYI: Stated main result to use as few acronyms as possible.}
\end{theorem}

% TB: I'd like a comment like this, but right now it feels clumsy so let's not
%Note that our algorithm has need for Holm et al.'s data structure only
%for the ``Case-1-contractions'', because for a Case-2-contraction (say
%of $z_2$ into $x$) we spend $\Theta(\deg(z_2))$ otherwise anyway and
%could have done the contraction in the same time even with incidence lists.

\section{Conclusion}
\label{conclusion}

In this paper, we reproved a result from Ackerman that all 1-planar graphs
admit a partition into a planar graph and a forest. Our proof is more general
than Ackerman's; the forest we find is guaranteed to not contain a path
between adjacent vertices of the input graph. Using this proof
and a data structure from Holm et al. for efficiently contracting the edges of
a planar graph, we were able to find 
%to establish Ackerman's claim that   % removed that, it makes our result sound weaker
this partition in linear time. 
In consequence, a number of results for 1-planar graphs (such as splitting
into 4 forests or 4-list-coloring if the graph is bipartite) can now be
achieved in linear time.
We also showed that the same algorithm can
be implemented in $O(n \log n)$ time with a simpler data structure that
uses only incidence lists.

As for open problems, the most interesting one is whether the partition
could be found even {\em without} being given the 1-planar drawing.
(Recall that it is NP-hard to find such a drawing \cite{1-planar-np-hard},
though it is polynomial for optimal 1-planar graphs~\cite{recognize-optimal-1-planar}.)
All papers listed in the introduction
for finding various edge partitions of 1-planar graphs require such an embedding.
%\todo{How about the other partition papers, do they require the embedding?
%I'm assuming they do, but if they don't then we need to rewrite this.
%SB: all require the embedding.}
%with Ackerman's and Czap and Hud\'ak's proofs, we required the 1-planar graph
%If so, then this should be phrased more strongly.}
%to be given with its embedding (though in the case of Czap and Hud\'ak, their
If this is difficult, could we at least do some of the implications
(such as splitting into 4 forests or orienting such that all in-degrees
are at most 4) in linear time without a given 1-planar drawing?

\bibliography{lifetime}

\appendix

\section{Omitted Proofs}

% this is a nasty hack to make the number match up, ideally we'd replace "3"
% by "value of "quadrangle:lemma:1" but that's more latex-fiddling than I feel 
% ready for right now.
\setcounter{theorem}{3}
%\addtocounter{theorem}{-1}
\begin{lemma}
%    \label{quadrangle:lemma:1}
    Let $H$ be a plane multigraph without loops, let $f$ be a quadrangle of $H$,
    and let $\angles{z_0,z_1,z_2,z_3}$ be the facial cycle of $f$.
    If $z_i=z_{i+2}$ (addition module 4) for some $i$,
    then $z_{i+1} \neq z_{i+3}$, and there is no face $f'\neq f$ which contains
    $z_{i+1}$ and $z_{i+3}$.
\end{lemma}
\begin{proof}
    (see Figure~\ref{contract:z1:z3}(b))
    Up to renaming, we may assume that $i=0$, so $z_0=z_2$.
    Assume by way of contradiction that $z_1=z_3$.
    Then $f$ consists of several parallel edges between $z_0=z_2$ and $z_1=z_3$, and thus
    $f$ is not a quadrangle.

    Assume by way of contradiction that there is some face $f' \neq f$ which
    contains $z_1$ and $z_3$. Subdivide one of the $z_0z_3$ edges to obtain
    a new vertex $y$ adjacent to $z_0$ and $z_3$, and further add an edge $yz_1$
    within $f$.
    We have split $f$ and attained a simple quadrangle $f''$ with facial cycle
    $\angles{z_0, z_1, y, z_3}$. Stellate $f''$ with a new
    vertex $c$, and add an edge $z_1z_3$ through the face $f'$.
    All these steps maintain the planarity of $H$. Moreover, the five
    vertices $z_0, z_1, y, z_3, c$ are pairwise adjacent. But this forms a $K_5$ which
    is not planar, a contradiction.
\end{proof}

\begin{lemma}
%    \label{quadrangle:lemma:2}
    Let $H$ be a plane multigraph without loops, let $f$ be a quadrangle of $H$,
    and let $\angles{z_0,z_1,z_2,z_3}$ be the facial cycle of $f$.
    If $z_i$ and $z_{i+2}$ (addition modulo 4) for some $i$ are both on some face
    $f' \neq f$, then no face $f'' \neq f,f'$ contains both $z_{i+1}$ and $z_{i+3}$.
\end{lemma}
\begin{proof}
    (see Figure~\ref{contract:z1:z3}(a,c,d,e))
    Up to renaming we may assume that $i=0$, so $z_0$ and $z_2$
    are on $f$ and $f'$. Suppose by contradiction that such a face $f''$ exists.
    By Lemma~\ref{quadrangle:lemma:1}, $z_1 \neq z_3$ and $z_0 \neq z_2$.
    Stellate $f$ with a new vertex $c$, add an edge $z_0z_2$ through $f'$,
    and add an edge $z_1z_3$ through $f''$. The original multigraph $H$ was planar,
    and all of these operations preserve planarity. However, the five vertices
    $z_0, z_1, z_2, z_3$, and $c$ are pairwise adjacent and form a $K_5$, which
    is not planar, a contradiction.
\end{proof}
\section{Omitted Pseudocodes}
\label{appendix-algorithms}

\begin{procedure}
    \KwResult{Contract two vertices $u$ and $v$ in $H$ through a quadrangle-face $f$,
%    given the quadrangle-vertex corresponding to $f$, 
	and return the resulting vertex $y$.}

    \tcp{pre: $f$ is labelled $\qu$}
    \tcp{pre: $v$ and $w$ are opposing on face of $H$ corresponding to $f$}

    Find the original edge $uv$ of $G$ that is stored with $f$.\\
    Record edge $uv$ as belonging to the forest of the partition.\\
    \For{$\angles{e, w}$ in $\texttt{neighbors}(f)$} {
        \If{$v$ equals $w$ \textbf{or} $v$ equals $u$ \textbf{or}
            $w$ has label $quad_i$ for some $0 \leq i \leq 3$} {
            $y := \texttt{contract}(e)$\\
        }
    }
    Remove labels $quad, quad_i$ from $y$\\
    \Return $y$\\
    \caption{ContractThrough($u, v, f$)}
    \label{contract}
\end{procedure}

\begin{procedure}
    \KwResult{Reconstruct the facial circuit of a quadrangle-face, given the
    corresponding quadrangle-vertex.}
    \tcp{pre: $f$ is a vertex of $H^\diamondsuit$ with label $quad$}

    \For{vertex $v$ in $\texttt{neighbors}(f)$} {
        \For{$i = 0, \ldots, 3$} {
            \lIf{$v$ has label $quad_i$} {
                $f_i := v$
            }
        }
    }
    \tcp{By construction $f_i$ has neighbours $\{z_{i-1},z_i,f\}$ (indices are mod 4)}
    \For{$i = 0, \ldots, 3$} {
        $N_i := \texttt{neighbors}(f_i) \cap \texttt{neighbors}(f_{i+1})$ \;
        \lIf{$|N_i| \textbf{ equals } 2$} { $z_i := N_i \setminus \{f\}$ }
    }
    \tcp{If $f$ is not simple then $|N_i|=3$ for two values of $i$, see Fig.\ref{gadget}. But then $z_i$ is determined since we know $z_{i-1}=z_{i+1}$ already}
    \For{$i = 0, \ldots, 3$} {
        \lIf{$|N_i| \textbf{ equals } 3$} { $z_i := N_i \setminus N_{i+2}$ }
    }
    \Return $\angles{z_0, z_1, z_2, z_3}$\\
    \todo[inline]{This should fix the algorithm for Figure 2(right).  TB: Reordered a bit (do the easy case first) and added more explanation.}
    \caption{FacialCycle($f$)}
    \label{reconstruct}
\end{procedure}

\begin{procedure}
    \KwResult{Adds all quadrangle-vertices incident to a vertex $y$ to a worklist,
    taking care not to put duplicates in the worklist.}
    \tcp{pre: $y$ is a vertex of $H$}
    \tcp{pre: $y$ equals $x$ (the vertex we currently work on), or $y$}
    will be contracted into $x$\\
    \For{vertex $v \in \texttt{neighbors}(y)$} {
        $v.\texttt{adj} := \textbf{true}$\;
        \If{$v$ has label $quad$ \textbf{and} \textbf{not} $v$.\texttt{in\_worklist}} {
            $v.\texttt{in\_worklist} := \textbf{true}$ \;
            $\texttt{worklist}.push(v)$\;
            $\angles{z_0, z_1, z_2, z_3}$ := \facialcycle{$v$}\;
            relabel $z_i$ such that $y$ equals $z_0$\;
            $z_2.\texttt{opposing}$ += 1\;
        }
    }

    \caption{InitializeAtOneVertex($y, \texttt{worklist}$)}
    \label{alg:add:quads}
\end{procedure}

\begin{procedure}
    \KwResult{Cleanup the metadata at a vertex $y$ and its neighbors.}
    \tcp{pre: $y$ is a vertex of $H$}
    \For{vertex $v \in \texttt{neighbors}(y) \cup \{y\}$} {
        $v$.\texttt{adj} := \textbf{false} \;
        $v.\texttt{in\_worklist} := \textbf{false}$ \;
        $v.\texttt{opposing} := 0$ \;
    }
    \caption{CleanupAtOneVertex($y$)}
    \label{alg:cleanup}
\end{procedure}

\begin{procedure}
    \KwResult{Find a PGF-partition of a graph $G$.}
    Add edges to make $G$ planar maximal 1-planar\;
    Compute planarization $G^\times$, mark vertices of crossings with $\qu$\; 
    \ForEach{vertex $f$ marked $\qu$}{
	Insert four vertices in four incident faces of $f$\;
        Mark these vertices with $\qu_0,\dots,\qu_3$ according to embedding\;
    }
    \While{there remains a vertex $f$ labeled $\qu$} {
        $x$ := some neighbor of $f$ not labeled $quad_i$\;
        \handlequads{$x$}\;
    }
    Return all edges that were recorded as forest $F$ and $G \setminus F$
    as planar graph
    \caption{FindPGFPartition($G$)}
    \label{final:alg}
\end{procedure}

\end{document}